\documentclass[x11names,final]{article}
\usepackage{arxiv}
\PassOptionsToPackage{x11names}{xcolor}
\usepackage{mathrsfs}
\usepackage{algorithm}
\usepackage{algpseudocode}
\usepackage{booktabs}
\usepackage[obeyFinal]{todonotes}
\usepackage{tikz}
\usepackage{graphicx}
\usepackage{subcaption}
\usepackage[normalem]{ulem}

\usepackage[inline]{enumitem}
\usepackage{balance}
\usepackage{hyperref}
\usepackage{amsmath,amssymb,amsthm}
\usepackage{url}
\usetikzlibrary{calc,fit,decorations.pathreplacing,shapes.geometric,arrows}
\tikzset{motion/.style={draw,line width=1.75pt},
  deviation/.style={motion,Firebrick3},
  continuation/.style={motion,Chartreuse3},
  announcement/.style={motion,DodgerBlue4},
  announcement overlay/.style={announcement, fill=DodgerBlue4!60,fill opacity=0.5},
  plan/.style={announcement,dashed},
  step font/.style={font=\sffamily\footnotesize},
  planStep/.style={motion,circle,step font,inner sep=1.5pt,fill=white, fill opacity=1},
  continuation step/.style={step font,black,opacity=0.5},
  diagram arrows/.style={line width=1.5pt,latex-latex},
  communication arrows/.style={diagram arrows, Sienna2},
  communication/.style={decorate, decoration={waves,segment length=4pt,angle=40,radius=5pt}, line width=1.5pt, Sienna2},
  observation/.style={line width=1.25pt,Sienna2,latex-latex}
}

\newcommand{\planSteps}[4][]{
  \foreach \s [count=\si from #3+1,remember=\si-1 as \siPrev (initially #3)] in {#4} {
    \path[#1,solid] (#2\siPrev) ++\s node[draw opacity={0.1*\si+0.3},planStep] (#2\si) {\si};
    \draw[#1,motion,opacity={0.1*\si+0.3}] (#2\siPrev) -- (#2\si);
  }
}

\tikzstyle{process} = [rectangle, rounded corners=1mm, text centered, draw=black, fill=orange!30]
\tikzstyle{decision} = [diamond, text width=3cm, inner sep=0, aspect=2, text centered, draw=black, fill=green!30]
\tikzstyle{arrow} = [->,>=stealth]

\tikzstyle{arena} = [step=1cm,Ivory4]
\tikzstyle{forbidden} = [Firebrick2]
\date{\today}

\newcommand{\RQ}[1]{\textbf{RQ#1}}

\theoremstyle{definition}
\newtheorem{definition}{Definition}

\newtheorem{theorem}{Theorem}
\newtheorem{lemma}{Lemma}

\newcommand{\paratitle}[1]{\noindent{\bf #1}.}

\newif\ifhidecomments
\hidecommentstrue

\title{HoLA Robots: Mitigating Plan-Deviation Attacks in Multi-Robot Systems with Co-Observations and Horizon-Limiting Announcements}
\renewcommand{\shorttitle}{Mitigating Plan-Deviation Attacks}
\renewcommand{\headeright}{preprint}
\renewcommand{\undertitle}{preprint}
\author{
    Kacper Wardega\\
    Boston University\\
    Boston\\
    USA\\
    \texttt{ktw@bu.edu}\\
    \And
    Max von Hippel\\
    Northeastern University\\
    Boston\\
    USA\\
    \texttt{vonhippel.m@northeastern.edu}\\
    \And
    Roberto Tron\\
    Boston University\\
    Boston\\
    USA\\
    \texttt{tron@bu.edu}\\
    \And
    Cristina Nita-Rotaru\\
    Northeastern University\\
    Boston\\
    USA\\
    \texttt{c.nitarotaru@northeastern.edu}\\
    \And
    Wenchao Li\\
    Boston University\\
    Boston\\
    USA\\
    \texttt{wenchao@bu.edu}\\
}

\begin{document}
    \maketitle
    
 \begin{abstract}

Emerging multi-robot systems rely on cooperation between humans and
robots, with robots following automatically generated motion plans to service
application-level  tasks.  
Given the
safety requirements associated with operating in proximity to humans and expensive infrastructure, it is important to understand and mitigate the
security vulnerabilities of such systems caused by compromised robots who
diverge from their assigned plans.
We focus on centralized systems, where a \emph{central entity} (CE) is responsible for
determining and transmitting the motion plans to the robots, which report
their location as they move following the plan.  The CE checks that
robots follow their assigned plans by comparing their expected location to
the location they self-report.  We show that this self-reporting monitoring mechanism is vulnerable to
\emph{plan-deviation attacks} where compromised robots don't follow their assigned plans while trying to conceal their movement by mis-reporting their
location.  We propose a two-pronged mitigation for plan-deviation attacks:
\begin{enumerate*} \item an attack detection technique leveraging both the robots' local sensing
    capabilities to report observations of other robots and {\em
        co-observation schedules} generated by the CE, and \item  a prevention%
    technique where the CE issues {\em horizon-limiting announcements} to the
    robots, reducing their instantaneous knowledge of forward lookahead steps
    in the global motion plan.\end{enumerate*}  
On a large-scale automated warehouse benchmark, we show that our solution enables attack prevention guarantees from a stealthy attacker that has compromised multiple robots.

\end{abstract}

\section{Introduction}

In this work we study attacks and defenses in multi-robot systems (MRS) following  a centralized execution model~\cite{honig2019persistent}, which is
representative of MRS in known, structured
environments with centralized management and control. 
The system consists of an external \emph{application},
the robots achieving the task, and a \emph{central entity} (CE) 
which is responsible for determining and transmitting the motion plans to each one of the robots.
Ideally, unplanned deviations due to malfunctions are detected by the CE by comparing the expected position of the robots to the one they self-report. 
Unfortunately, compromised robots who deviate from the  motion plan and attempt to move through forbidden regions of the environment cannot be detected solely
by self-reports of location from robots, as the compromised ones can lie in their reports to remain undetected.
We refer to such deliberate deviations as  \emph{plan-deviation attacks} and we focus on them in this work.

Plan-deviation attacks were previously introduced in
\cite{wardega_resilience_2019}, which proposed to use co-observations of
other robots to detect deviations. Specifically, logic-based planning centered around formal specification of the detection constraint result in motion plans such that the implied co-observation
schedule can guarantee detection for {\em a single compromised robot}. However, such plans
are not guaranteed to exist, and the intractability of the logic-based planning problem prevents the approach from scaling to realistic MRS deployments.
More importantly, \cite{wardega_resilience_2019} does not generalize to multiple compromised robots.
We design our solution to address these concerns based on two observations about the attackers:
\begin{enumerate*} \item they use the motion plan information from the CE to
            determine how to move towards the forbidden zone, and \item they
lie about their location to try to remain undetected by the CE \end{enumerate*}. 
The key idea of our approach is a novel mechanism
of {\em horizon-limiting announcements} (HoLA), where we limit how much motion
planning information is announced to the robots at any given time in order to
stymie the ability of the attacker to plan successful attacks, but still send
as many steps as possible. This is achieved through an \textit{efficient} verification algorithm conducted by
the CE which checks whether the planned announcements prevent stealthy attackers from
moving towards the forbidden zone because of not having enough
information; in the worst case only one step will be released. 
In this work, our contributions are:
\begin{itemize}

\item We provide a formal characterization of plan-deviation attacks, centered around {\em stealthy attackers} who deviate from the plan only if they know that 
they can move towards the forbidden region while remaining undetected.

\item We propose a mitigation, HoLA, for plan-deviation attacks that combines 
co-observation schedules with issued horizon-limiting announcements to prevent
attacks from stealthy attackers.

\item We provide formal guarantees that horizon-limiting announcements prevent attacks from a stealthy attacker that has compromised multiple robots.

\item We propose a procedure for efficiently computing the maximum-length horizon-limiting announcements. We evaluate the computation overhead of the verification
and show that the procedure
scales well to instances with many robots; the procedure exhibits
robot-level parallelism and takes no more than 2 minutes running  on a single core %
to verify scenarios with 100 robots.

\end{itemize}

\section{Problem Formulation}
\label{sec:formulation}

We focus on the centralized MRS model which consists of a set of robots ($R$), and a \emph{central entity}~(CE) that communicates with and manages the robots. %
The~CE accepts as input a queue of application tasks that are to be carried out
by the robots in the environment, computes
multi-robot motion plans, $x$, that carry out the application tasks, and then iteratively announces portions of the
motion plans, $\alpha(t)$, to the robots.  The CE ensures that the motion plans adhere to
safety constraints in the form of locations in the environment that are marked
as out-of-bounds to the robots. These could be due to a variety of reasons,
e.g. a human moving through the
environment, robots experiencing localization faults, unsafe conditions in the
environment, etc.
The  environment is modeled %
as a graph~$G=(V,E)$, with time-varying out-of-bounds locations denoted $V_\text{forbidden}(t)\subset V$.   Motion plans in the
centralized MRS model are formally defined as follows~\cite{stern_multi-agent_2019}.

\begin{definition}[MAPF plan]
	\label{def:mapf-plan}
	
	A multi-robot path-finding plan for robots~$R$ in the environment~$G=(V,E)$
	is a finite sequence~$\{ x_t\}$ with elements~$x_t\in
	V^R$, where the sequence~$x^i=\{x_t^i\}$ is the single-robot plan for robot~$i\in R$, and that satisfies the following constraints for all $t$ and for all~$i,j\in R$:
	\begin{enumerate*}
		\item Each~$x^i$ is a walk on~$G$.
		\item robots do not occupy the same location simultaneously.
		\item robots do not traverse the same edge simultaneously. 
	\end{enumerate*}
\end{definition}

The announcements made by the CE are MAPF prefixes, i.e.~$\alpha(t)\preceq x$, defined as follows.
	\begin{definition}[MAPF prefixes and continuations]
		\label{def:mapf-prefix}
		
		Let~$x$ and~$y$ be two MAPF plans. We say that~$y$ is a \emph{MAPF prefix}
		of~$x$ and equivalently that~$x$ is a \emph{MAPF continuation} of~$y$,
		denoted as~$y\preceq x$, if~$y^i$ is a prefix for~$x^i$ for all~$i\in R$.

	\end{definition}  

\paratitle{Attacker model} Assume that an attacker has compromised a subset~$A\subseteq R$ of
the robots, with the intention to sabotage the system and cause robots in $A$ to violate the CE's safety constraints without being detected. The compromised robots have full information of the motion plan announcements $\alpha(t)$ from the CE, however the compromised robots do not know which other robots are compromised and are unable to coordinate, and
hence need to act independently.  We exclude strong-coordination between
attackers because  in centralized settings, in-protocol communication between
robots is monitored, and the robots are moving in an area where such network
communication would be detected. We also assume that the robots do not have
access to other side-channels for communication.  
Malicious deviations from the nominal plan conducted by a compromised robot
are not easily detectable by the~CE, since the compromised robot can lie in its
self-reports to the~CE.
We refer to such malicious deviations as \emph{plan-deviation attacks}, and
to deviations that in addition seek to move the robot into one of the forbidden areas
in $V_\text{forbidden}(t)$ as \emph{forbidden plan-deviation attacks}. We formalize these threats below.

\begin{definition}[Plan-Deviation Attack]
	\label{def:mapf-dev}
	
	Let $x$ be a MAPF plan for set of robots $R$ on map $G=(V,E)$. We say that
	$\tilde{x}$ is a MAPF deviation for robot $i\in R$ on timesteps $(s,f)$ from $x$ if
	$\tilde{x}$ satisfies $(\forall
	j,t)(x_t^j\neq\tilde{x}_t^j\Leftrightarrow(j=i, s<t<f))$. 
\end{definition}

\begin{definition}[Forbidden Plan-Deviation Attack]
	\label{def:dang-mapf-dev}
	
	A MAPF deviation $\tilde{x}$ for robot $i$ on $(s,f)$ is a \textit{forbidden}
	deviation, in short $\mathcal{F}(\tilde{x},x,i,s,f)$, if
	\((\exists t\in(s,f))\) s.t. \((\tilde{x}_t^i\in V_\text{forbidden}(t))\).
	
\end{definition}

\begin{figure}[t]
	\centering
	\begin{tikzpicture}[scale=1.0]
		\draw[arena] (-1,0) grid (7,5);
\foreach \x/\y in {-1/3,0/1,1/1,3/1,3/4,5/0,5/2,5/3} {
  \fill[black] (\x,\y) rectangle (\x+1,\y+1);
}
\fill[forbidden] (1,2) rectangle (2,3);
		\coordinate (j area) at (2,0);
		\draw[announcement overlay] (j area) -- ++(3,0) -- ++(0,1) -- ++(-2,0) -- ++(0,2)%
		-- ++(-1,0) -- ++(0,-2) -- ++(-2,0) -- ++(0,-1) -- cycle;
		\coordinate (k area) at (6,1);
		\draw[announcement overlay] (k area) -- ++(-1,0) -- ++(0,-1) -- ++(-1,0) -- ++(0,3)%
		-- ++(1,0) -- ++(0,-1) -- ++(1,0) -- ++(0,1) -- ++(1,0) -- ++(0,-3) -- ++(-1,0) -- cycle;
		\input{figures/tikz/robotsAnnouncements}
		\node at ($(kA0)+(2mm,2mm)$) {\includegraphics[width=4mm]{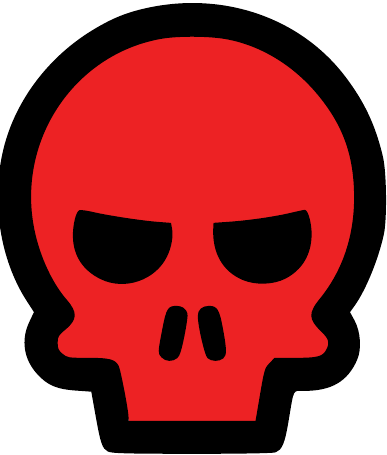}};
		\foreach \x/\y in {1.5/0.5,0.5/1.5,2.5/1.5} {
			\node[continuation step] at ($(j area)+(\x,\y)$) {4};
		}
		\foreach \x/\y in {2.5/0.5,0.5/2.5,2.5/2.5,4.5/0.5} {
			\node[continuation step] at ($(j area)+(\x,\y)$) {5};
		}
		\input{figures/tikz/deviation}
	\end{tikzpicture}
	\caption{
		The compromised robot~$i$ 
		has computed a forbidden MAPF deviation~$\tilde{x}$ 
		(red paths) on timesteps~$(1,5)$. A stealthy attacker,
		however, realizes that there is a \emph{possible continuation} (shaded blue region) from the announced
		portion of the~CE's MAPF plan (blue paths) that would result in a
		co-observation-based detection by the~CE: if robot~$j$ goes north at time
		step 3, then~$j$ would observe~$i$ at a location where~$i$ is not
		supposed to be. As a result, the stealthy attacker chooses not to perform the plan-deviation attack. 
	}
	\label{fig:attacker types}
\end{figure}

\paratitle{Undetected plan-deviations}
Assume that, up to time~$t$, no plan deviation attack has
been attempted, and so the true system state~$\tilde{x}_t$ matches the~CE's
expectation~$x_t$. A compromised robot~$a\in A$ may choose to deviate from the
plan by picking a different action~$(x_t^a,\tilde{x}_{t+1}^a)\in E$ s.t.
$\tilde{x}_{t+1}^a\neq x_{t+1}^a$. In order to hide that the deviation has
occurred, the compromised robot would falsify its self-report and attest to the CE that
it has moved into the nominal location. Provided that $a$ has
not collided with a non-compromised robot, i.e.~$\tilde{x}$ is still a MAPF
plan, and that $a$ has not caused a non-compromised robot~$i\neq a$ to be
unable to perform an action, i.e. that~$\tilde{x}$ is a MAPF deviation for~$i$
from $x$, then it is easy to see that none of the self-reports from the
robots will have changed. Such plan deviations are called \emph{undetected plan deviations}.

\paratitle{Stealthy attackers} This type of attacker uses their knowledge of 
the currently announced MAPF prefix $\alpha(t)$ to determine whether
there exists a MAPF plan $\tilde{x}$ that is guaranteed to be a forbidden
undetected deviation from the true plan $x$, $x\succeq\alpha(t)$.
Specifically, a stealthy attacker needs to ensure that there is a MAPF continuation
$\tilde{x}$ from $\tilde{x}_{t}$ s.t. $\tilde{x}$ is a forbidden \emph{and} undetected MAPF
deviation from $x$ on $(t,f)$ prior to actually executing the deviation. In practice, the attacker can easily
verify this if it has enough information about $x$; if the announcement
$\alpha(t)$ reveals a large horizon of the plan, the stealthy attacker $a$ can easily
solve a single-robot planning problem \cite{choset2005principles} using $\alpha(t)$ to avoid conflicts with
the other robots $i\neq a$.

\paratitle{Security-aware execution problem} 
For a variety of reasons, the~CE wants to announce as much of the MAPF plan as
possible, e.g. due to considerations for network latency, contention, or
robustness to network and motion faults~\cite{atzmon_robust_2020}. Hence,  at each time $t$, the CE aims to \emph{maximize} $|\alpha(t)|$ subject to the constraint that the unknown compromised subset $A\subseteq R$ of stealthy attackers are not able to perform forbidden plan-deviation attacks.

\section{Mitigating Plan-Deviation Attacks} %
\label{sec:solution}

In this section we present a solution to plan-deviation attacks
against centralized MRS.  Our approach consists of two core components, co-observations and horizon-limiting announcements.

\subsection{Co-observation Schedules}
\label{subsec:coobs-sched}

In order to decrease the set of MAPF deviations that go undetected by the CE,
we propose to include {\em co-observations} of other robots in the self-reports sent to
the CE. Ordinarily, the onboard sensing capabilities of the robots are only
used to avoid collisions in fault scenarios. However, we notice that using the
sensors to report all inter-robot observations has measurable benefits for
security. 

Our approach is to include in robot~$i$'s self-report at time~$t$,
$\tilde{\beta}(t)^i$, all observations that~$i$ makes of other robots at time~$t$, in
addition to~$i$'s self-report on action success. As an example, say that robot
$i$ is at location~$v$ and robot~$j$ is at location~$w$, and~$(v,w)\in E^*$ (in
other words~$i$ can observe~$j$ from its vantage point). Then
$\tilde{\beta}(t)^i=\{\tilde{x}_{t+1}^i=x_{t+1}^i,\tilde{x}_t^j=w\}$, or in
plain English, ``$i$ reports that~$i$ has moved successfully to~$x_{t+1}^i$ and
that~$i$ observed~$j$ at time~$t$ at location~$w$.'' We note that this generalizes straightforwardly to environments instrumented with fixed observers (cameras) or fully-trusted agents. 

\begin{definition}[Co-Observation-Based Detection]
    \label{def:co-observation-detection}

    Let~$x$ be a MAPF plan and~$\beta$ be the localization and co-observation
    self-reports implied by successful execution of~$x$:
    $$\beta(t):=\{\{\tilde{x}_{t+1}^i=x_{t+1}^i\}\cup\{(i,j,\tilde{x}_t^j):j\in
    R\setminus i\land(\tilde{x}_t^i,\tilde{x}_t^j)\in E^*\}\}_{i\in R}$$ If any
    robot~$i\in R$ fails to perform an action, does not observe a robot that it
    should have, or does observe a robot that it should not have, then
    the self-report~$\tilde{\beta}(t)^i$ sent by~$i$ to the CE will not 
    match~$\beta(t)^i$, triggering a co-observation-based detection in the CE.

\end{definition}

\subsection{Horizon-Limiting MAPF Announcements}
\label{subsec:secure-ann}

We now
focus on making the attack planning problem against a system with robot
co-observation-based mitigation more difficult given a general MAPF plan.  \emph{The key idea is as follows:  the CE can improve the security of the system by
preventing the attacker from easily computing forbidden and undetected
plan-deviation attacks.} The simplest way to accomplish this is to limit the
amount of information available to the attacker about the MAPF plan, that is,
by limiting the amount of future planning information available at every time
instant, $\alpha(t)$.

\noindent
\paratitle{Limiting stealthy attackers}
Consider again the attack planning problem for a stealthy attacker~$a\in A$.
Since the stealthy attacker only attempts a plan-deviation attack if success
and stealth are ensured, the amount of information that the attacker has about
the plan is critical -- if~$\alpha(t)$ provides planning information on a long horizon, the attack planning problem is essentially a \textit{graph reachability}
problem.  Formally, this is the case when there exists a forbidden, undetected
deviation for~$a$ on~$(t,f)$ where~$f$ is less than the length of the shortest
single-agent plan in~$\alpha(t)$, i.e.~$f<\min_i\lvert\alpha(t)^i\rvert$.  If~$\alpha(t)$
does not reveal so much information, however, the attack planning problem is
made considerably more difficult. This is because the attacker needs to compute
a deviation that is not only forbidden, but also guaranteed to be undetected
for all possible MAPF continuations of~$\alpha(t)$. Conversely, this tells us
that to mitigate attacks from stealthy attackers, it suffices to show that for
every forbidden deviation for~$a$ from~$x$ that there exists a continuation
from~$\alpha(t)$ would result in a detection, in which case the stealthy
attacker would abstain from deviating from the plan. This motivates a class of
announcement strategies for the MAPF plan~$x$ that guarantees security from
stealthy attackers:

\begin{definition}[Horizon-Limiting MAPF Announcements]
\label{def:secure-mapf-announcements}
 Let~$x$ be MAPF plan on~$G$ for~$R$,~$\alpha$ an announcement sequence for~$x$, and~$\beta_x$
 the sequence of robot self-reports implied by~$x$. Then~$\alpha$ are horizon-limiting MAPF announcements for~$x$ iff

        $$\label{eq:verify}
        (\forall\tilde{x},i\in R,t,f\in\mathbb{N})(\exists y)(\mathcal{F}(\tilde{x},x,i,t,f)\Rightarrow y\succeq\alpha(t)\land\beta_y^{R\setminus\{i\}}\neq\beta_{\tilde{x}}^{R\setminus\{i\}})$$

    That is, the announcements~$\alpha$ are considered horizon-limiting if and only if
    they at no point reveal enough information for the attacker to be certain
    that a given forbidden MAPF deviation will be undetected by the CE, since
    there exists some continuation~$y$ from~$\alpha(t)$  such that the
    self-reports induced by~$y$ do not match the self-reports induced by the
    deviation. For a given time $t$, we say that $\alpha(t)$ is the \emph{maximum-length horizon-limiting announcement} if there does not exist any $\alpha^*(t)$ s.t. $|\alpha(t)|<|\alpha^*(t)|$, $\alpha^*(t)\preceq x$, where $\alpha^*(t)$ is horizon-limiting.

\end{definition}

\begin{theorem}[Guaranteed Security from stealthy Attackers]
    \label{thm:stealthy}

Let~$x$ be a MAPF plan and assume that the CE uses a horizon-limiting MAPF 
announcement~$\alpha$ for~$x$. Then no robots compromised by a stealthy attacker would attempt a plan-deviation attack.

\end{theorem}

\subsection{Synthesis of Horizon-Limiting Announcements}
\label{subsec:verif}

In our approach, the CE first leverages a conventional,
non-security-aware MAPF solver in order to compute a cost-optimized MAPF plan
as it would in a typical deployment. In the post-processing step however, we
verify that Eq.~\ref{eq:verify} holds before fixing the announcements $\alpha$.
If the announcements cannot be verified to be horizon-limiting, then we attempt
to resolve the issue by iteratively choosing less-informative announcements
until the maximum-length horizon-limiting announcement is found.

The main challenge that we face in designing our verification procedure is the computational complexity of MAPF itself, which is known to be NP-hard~\cite{yu2013structure}. Therefore, a complete attack planning algorithm for the stealthy attacker with
imperfect information is computationally difficult as it entails enumerating MAPF continuations.
This motivates us to
instead focus on developing an incomplete, but sound and efficient, verification procedure
for the horizon-limiting announcement checking problem, Eq.~\ref{eq:verify}. Our solution
is \emph{non-deterministic co-observation enumeration}, shown in
Alg.~\ref{alg:verify}. We base our algorithm on an abstraction of MAPF planning
that allows non-deterministic movements for the robots on $G$. Our abstraction
allows non-compromised robots to ignore vertex and edge constraints of MAPF
plans, allowing us to efficiently explore the co-observation schedules of many
MAPF continuations from the current $\alpha(t)$ simultaneously.  Although our
abstraction does over-approximate the set of MAPF continuations, we can prove
that the abstractions preserve the possibility of pairwise co-observation.
That is, if under the abstraction it is possible for a robot $i$ to observe
robot $j$ at a location $v$ at time $t$, then there is some MAPF continuation
where $j$ is observed at location $v$ at time $t$. This property of the
abstraction ensures that Alg.~\ref{alg:verify} is sound, since
Alg.~\ref{alg:verify} is essentially verifying that there is no forbidden
deviation through the complement of the observed region under the
non-deterministic movement abstraction.

\algblockdefx[DoBlock]{Do}{EndDo}
{\textbf{do}}
[1][true]{\textbf{while} #1}

\begin{algorithm}[h!]
	\caption{Non-deterministic Co-observation Enumeration}
	\label{alg:verify}
	\begin{algorithmic}[1]
		\small
		\Procedure{Verify}{$G,S,V_\text{forb.},\alpha,s,a$}
		\State $u\gets s$\Comment{time offset}
		\State $X_u^R\gets\alpha(s)_u^R$\Comment{init. reachable sets for each robot}
		\Do
		\State $X_{u+1}^R\gets\Call{Reachable}{\alpha(s),G,X_{u}^R,u,C}$
		\State $X_{u+1}^a\gets X_{u+1}^a\setminus X^{R\setminus\{a\}}_{u}$
		\State $X_{u+1}^{R\setminus\{a\}}\gets X_{u+1}^{R\setminus\{a\}}\setminus X_{u+1}^a$
		\State $X_{u+1}^a\gets X_{u+1}^a\setminus X_{u+1}^{R\setminus\{a\}}$
		\State $u\gets u+1$
		\EndDo[$X_{u}^a\cap\mathcal{N}_S(X_{u}^{R\setminus\{a\}})\neq\{\}$]
		\State $u^*\gets u$
		\State $Q\gets X_{u^*}^a\cap\mathcal{N}_G(X_{u^*}^{R\setminus\{a\}})$
		\State\Return $\bigvee_{q\in Q}\neg\Call{AttackExists}{a,G,V_\text{forb.},X,s,u^*,q}$
		\EndProcedure
		\Procedure{Reachable}{$x,G,X,t,C$}
		\State $X_\text{next}\gets\{\}$
		\For{$v\in X$}
		\State $X_\text{next}\gets X_\text{next}\cup\Call{MoveRobot}{x,G,v,t,C}$
		\EndFor
		\State\Return $X_\text{next}$
		\EndProcedure
		\Procedure{MoveRobot}{$x,G,v,t,C$}
		\If{$\exists r\in R, v=x_t^r\land|x^r|>t+1$}
		\State \Return $\{x_{t+1}^r\}$\Comment{prefix for $r$ is known}
		\EndIf
		\State $\text{ret}\gets\mathcal{N}_G(v)\setminus\{x_{t+1}^r:r\in R\land|x^r|>t+1\}$
		\If{$v\notin\text{ret}$}\Comment{avoid edge conflict}
		\State $\text{ret}\gets\text{ret}\setminus\{x_t^r:r\in R\land x_{t+1}^r=v\}$
		\EndIf
		\State $\text{ret}\gets\text{ret}\setminus\{v:(v,t+1)\in C\}$
		\If{$|\text{ret}|=0$}
		\State $C\gets C\cup\{(v,t)\}$
		\State $\textbf{raise } \texttt{CONFLICT}$
		\EndIf
		\State \Return $\text{ret}$
		\EndProcedure
		\Procedure{AttackExists}{$a,G,V_\text{forb.},X,s,u^*,q$}
		\State $\text{A}\gets X_t^a$
		\State $\text{B}\gets\{\}$
		\For{$u=s+1,\ldots,u^*$}
		\State $\text{A}\gets\mathcal{N}_G(A)\setminus\mathcal{N}_S(X_u^{R\setminus\{a\}})$
		\State $\text{B}\gets\mathcal{N}_G(B)\setminus\mathcal{N}_S(X_u^{R\setminus\{a\}})$
		\State $\text{B}\gets\text{B}\cup(\text{A}\cap V_\text{forb.})$
		\EndFor
		\State\Return $q\in B$
		\EndProcedure
	\end{algorithmic}
\end{algorithm}

The input to Alg.~\ref{alg:verify} is the centralized MRS instance $G=(V,E)$, $S=(V,E^*)$,
$V_\text{forbidden}(t)$, and the sequence of announcements planned by the CE
from the current time $t$ to a future time $f$, $\{\alpha(s)\}_{s\in[t,f]}$. We
iteratively fix each robot $a\in R$ as the compromised robot; by attacker
independence, from $a$'s perspective the other $R\setminus a$ may all be
non-compromised. We now iterate over the $s\in[t,f]$ and attempt to verify that
$\alpha(s)$ is not informative enough to reveal a forbidden and undetected
plan-deviation attack for robot $a$ beginning at time $s$.
Verifying $\alpha(s)$ has two phases: (1) compute the soonest time $u^*>s$ and a
location $l_\text{obs}$ where $a$ could be observed by a robot in
$R\setminus\{a\}$ and (2) show that no forbidden undetected deviation exists for
$a$ on $(s,u^*)$.

Since $\alpha(s)$ only reveals partial planning information for $i\in R$ up to
time $|\alpha(s)^i|$, we account for the unknown future of a given robot by
allowing them to move non-deterministically on $G$ for time steps
$u>|\alpha(s)^i|$. We denote the set of locations that $i\in R$
(non-)deterministically occupies at time $u$ as $X_u^i$. The dynamics of the
non-deterministically-moving agents are as follows:

\begin{enumerate}

    \item For all $i\in R$, $X_u^i=\{\alpha(s)_u^i\}$ for $u\leq|\alpha(s)^i|$, i.e. robots
        move deterministically for times where their position is specified by
        $\alpha(s)$.

    \item For $u>|\alpha(s)^i|$, $X_u^i\leftarrow\mathcal{N}_G(X_{u-1}^i)$,
        i.e. non-deterministically-moving robots follow all edges in $G$ from
        the set of locations previously occupied.

    \item For $u>|\alpha(s)^i|$, remove from $X_u^i$ all locations that are
        deterministically occupied by other robots $R\setminus\{i\}$, or would
        lead to a vertex- or edge-conflict with a deterministically-moving
        robot in $R\setminus\{i\}$. The conflict locations are stored in a set
        $C$, which is updated with a new conflict whenever there is a $c\in
        X_{u-1}^i$ that has no children (available actions) to $X_u^i$. The
        verification for $\alpha(s)$ is restarted whenever a new conflict is
        found.

    \item The non-deterministically-moving compromised robot $a$ cannot move
        into any location \emph{previously} occupied by non- deterministically moving
        robots in $R\setminus\{a\}$, so remove from $X_u^a$ all elements also in $X_{u-1}^{R\setminus\{a\}}$.

    \item Non-deterministically-moving robots in $R\setminus\{a\}$ cannot move
        into any location occupied non-deterministically by $a$, so remove from $X_u^{R\setminus\{a\}}$ all elements also in $X_u^a$.

    \item The non-deterministically-moving compromised robot $a$ cannot move
        into any location occupied by non-deterministically-moving robots in
        $R\setminus\{a\}$, so remove from $X_u^a$ all elements also in $X_u^{R\setminus\{a\}}$.

\end{enumerate}

The non-deterministic dynamics are evolved for $u=s+1,\ldots,u^*$, where $u^*$
is the first time step s.t. $\exists l_\text{obs}\in X_{u^*}^a$ s.t.
$l_\text{obs}\in\mathcal{N}_S(X_{u^*}^{R\setminus\{a\}})$, i.e. when a possible
observation on $a$ by another robot in $R\setminus\{a\}$ is found, concluding
the first phase of verifying $\alpha(s)$. For the second phase, we simply check
via graph search on $G$ from source vertex $\alpha(s)_s^a$ %
if there is a MAPF
deviation $\tilde{x}$ for $a$ on $(s,u^*)$ s.t. for all $u\in(s,u^*)$,
$\tilde{x}_u^a\notin\mathcal{N}_S(X_u^{R\setminus\{a\}})$. If no such deviation
is found, then we return \texttt{true}, signifying that there exists a
continuation from $\alpha(s)$ s.t. no forbidden undetected MAPF deviation
exists for $a$ on $(s,u^*)$ (in that continuation). If each $\alpha(s)$ is
verified for each $i\in R$, then the announcements $\{\alpha(s)\}_{s\in[t,f]}$
are verified to be horizon-limiting MAPF announcements.

\begin{theorem}[Soundness of Non-Deterministic Co-Observation Enumeration]
    \label{thm:sound}

    Let $x$ be a MAPF plan and $\alpha$ an announcement sequence for $x$. Then
    if Alg.~\ref{alg:verify} returns $\texttt{true}$, then $\alpha$ is a horizon-limiting
    MAPF announcement for $x$.

\end{theorem}

\begin{proof}
    \label{appx:verify-proof}
    Eq.~\ref{eq:verify} is equivalent to the statement that for all forbidden MAPF
deviations $\tilde{x}$ for $a\in R$ on $(s,f)$, there exists a MAPF
continuation $y$ of $\alpha(s)$ s.t. execution of attack $\tilde{x}$ would
trigger a co-observation-based detection if $y$ is the CE's MAPF plan. Let
$\{X_u\}_{u>s}$ be the sequence of (non-) deterministically reachable sets for the robots
starting at time $s$ as computed by Alg.~\ref{alg:verify}. We begin by proving a lemma that the non-deterministic movement abstraction
of Alg.~\ref{alg:verify} is sound w.r.t. possibility of co-observation:

\begin{lemma}[Non-deterministic Abstraction Preserves Possibility of Co-observations]
    \label{lem:sound-co-obs}

    Let $q\in X_u^a$. If $q\in\mathcal{N}_S(X_u^{R\setminus\{a\}})$, then there
    exists a MAPF continuation $y$, $y\succeq\alpha(s)$ s.t.
    $y_u^a\in\mathcal{N}_S(y_u^{R\setminus\{a\}})$. In other words, if it is
    possible under the non-deterministic abstraction for $a$ to be observed at
    time $u$ at location $q$, then there exists a MAPF continuation from
    $\alpha(s)$ where $a$ is observed at time $u$ at location $q$.

\end{lemma}

\emph{Proof of Lem.~\ref{lem:sound-co-obs}}: Firstly, since for all $i\in R$ there are no elements of $X_u^i$
that are not in $\mathcal{N}_G(X_{u-1}^i)$, we have that all locations in
$X_u^i$ are reachable in one time step by taking an edge in $E$ from some
location in $X_{u-1}^i$. Furthermore, since elements in $X_u^i$ are not in the
conflict set $C$, we have that there is a conflict-free walk on $G$ from
$\alpha(s)_s^i$ to each element in $X_u^i$ w.r.t. the known prefix $\alpha(s)$.
Since non-deterministically-moving $a$ is not allowed to move into any element
in $X^{R\setminus\{a\}}$, we therefore have that for all $q\in X_u^a$, there
exists a $y\succeq\alpha(s)$ s.t. $y_u^a=q$. As for
non-deterministically-moving pairs of other robots in $R\setminus\{a\}$, the
abstraction does not explicitly prevent vertex- and edge- conflicts. However,
since elements in $X_u^{R\setminus\{a\}}$ are not in the conflict set $C$, we
have that it is possible for some robot $i\in R\setminus\{a\}$ to occupy each
element of $X_u^{R\setminus\{a\}}$ without causing a conflict with
deterministically-moving robots. Similarly, since non-deterministically-moving
robots in $R\setminus\{a\}$ are not allowed to move into any element in $X^a$,
we conclude that the only conflicts preventing a robot $i\in R\setminus\{a\}$
from reaching a location in $X_u^i$ is a conflict with a different
non-deterministically-moving robot $j\in R\setminus\{a,i\}$. Therefore, for all
$p\in X_u^i$, either there exists a $y\succeq\alpha(s)$ s.t. $y_u^i=p$ or there
exists a $y\succeq\alpha(s),j\in R\setminus\{a,i\}$ s.t. $y_u^j=p$. We conclude
that for all $q\in X_u^a,p\in X_u^{R\setminus\{a\}}$, there exists a
$y\succeq\alpha(s)$, $i\in R\setminus\{a\}$ s.t. $y_u^a=q$ and $y_u^i=p$.
$\blacksquare$

\emph{Proof of Thm.~\ref{thm:sound}}: Now let $u^*$, $l_\text{obs}$ be the time
and position of the first possible observation on $a$ as computed by
Alg.~\ref{alg:verify}.  Let $\tilde{x}$ be any forbidden MAPF deviation for $a$
on $(s,s+k)$, $k>1$. \emph{Case I, $\tilde{x}_{u^*}^a\neq l_\text{obs}$}: it
follows immediately from Lem.~\ref{lem:sound-co-obs} that $\exists
y\succeq\alpha(s),i\in R\setminus a$ s.t.
$l_\text{obs}\in\mathcal{N}_S(y_{u^*}^i)$. Therefore, $a$ misses an
observation, triggering a co-observation-based detection in the CE.  \emph{Case
II(a), $\tilde{x}_{u^*}^a=l_\text{obs}$ and $\exists u\in(t,u^*)$ s.t.
$\tilde{x}_u^a\in\mathcal{N}_S(X_u^{R\setminus\{a\}})$}: in this situation, it
again follows immediately from Lem.~\ref{lem:sound-co-obs} that $\exists
y\succeq\alpha(s),i\in R\setminus\{a\}$ s.t.
$\tilde{x}_u^a\in\mathcal{N}_S(y_u^i)$. However, since $u<u^*$ we contradict
that the first possible observation on $a$ occurs at time $u^*$. Therefore, $a$
has caused an unexpected observation, triggering a co-observation-based
detection in the CE.  \emph{Case II(b), $\tilde{x}_{u^*}^a=l_\text{obs}$ and
$(\neg\exists u\in(t,u^*)$ s.t.
$\tilde{x}_u^a\in\mathcal{N}_S(X_u^{R\setminus\{a\}})$}: the only remaining
forbidden deviations are those where $a$ does not miss the planned observation
at time $u^*$, and does not introduce an unexpected observation at times
$u\in(t,u^*)$.  The algorithm performs a graph search to check that no such
deviation exists.  $\blacksquare$

\end{proof}

The computational complexity of Alg.~\ref{alg:verify} is $\mathcal{O}(RV)$, as the procedure terminates once the attacker's reachable set intersects one of the defenders' reachable sets -- a total of $R$ sets each with maximum cardinality $V$. Each potential attacker $a$ and announcement $\{\alpha(s)\}_{s\in[t,f]}$ can be verified in parallel, allowing for efficient computation of the maximum-length horizon-limiting announcement.

\section{Experimental Results}
\label{sec:exp}

In the preceding section, we have proposed a strategy for mitigating
plan-deviation attacks that rests on robot co-observations and on limiting how
much planning information is revealed at any given moment.  Here, we seek to
answer the following research questions:%
\begin{description}

\item[RQ1] What is the security benefit of HoLA, compared to a centralized MRS that detects problems using localization self-reports only?

\item[RQ2] Compared to a non-security-aware centralized MRS, what is the
overhead of HoLA?

    \item[RQ3]
What properties of robot co-observations from
general MAPF plans lead to security vulnerabilities?
 
    \item[RQ4] What is the security benefit for robot co-observation? How does the inclusion of robot co-observations impact our ability to mitigate plan-deviation attacks without using horizon-limiting announcements?

    \item[RQ5] How does the announcement schedule for incremental plans
    impact the effectiveness of attacks?

    \item[RQ6] What is the security vulnerability associated with announcement
        schemes that may be used in typical centralized MRS deployments?

\end{description}

\subsection{Experimental Setup}
\paratitle{Environment} MAPF plans are computed using the ECBS
algorithm~\cite{barer_suboptimal_2014}, an efficient and bounded sub-optimal
graph-based MAPF solver (and so, applicable for centralized MRS), for a set
of 100 standard MAPF 4-connected grid benchmark instances~\cite{libmultirobotplanning2021}. 
The MAPF instances are solvable (i.e.
there exists a MAPF plan that solves the instance), randomly generated
4-connected $32\times 32$ grids   with either 10, 20,~\ldots, or 100 robots and
$\sim$200 obstacles. We assume each robot has sensing capability
within adjacent squares, that is the sensor model for each robot, $S$, is the
same as the reachability graph $G$. Robots are assumed to
mutually co-observe each other if they are adjacent on the grid.   We implement the announcement security verification in the
Rust programming language; runtimes are reported on an Intel Core i7-6700
processor at 4GHz. Source code to reproduce our experiments can be found at \url{https://github.com/gitsper/hola-announce}

\noindent
\paratitle{Execution scenarios} There are two factors that influence the
announcement schedules: \begin{enumerate*} \item how many steps ahead are included in the
announcement and \item how many announcements are sent in a communication from the
CE to the robots.\end{enumerate*} The number of steps ahead represent a trade-off between
security and delay in computing the task, for increased security the
announcement should include only one step but this will results in increased time
in completing the task by the team of robots. We use the following
notation:

\begin{itemize}
\item $(p, k)$-announcements: specify that the CE makes a new announcement every $p$ timesteps and each announcement includes planning information for the next $k$ steps.
\end{itemize}

\paratitle{Stealthy attacker metrics} We consider the security of
a nominal execution scenario compromised by a stealthy attacker. Instead of
performing simulations, for each scenario we randomly pick one of 10 different
robots to play the role of stealthy attackers and one of 10 locations in the grid to
be marked as the forbidden location and use Alg.~\ref{alg:verify} to check if
the announcements are horizon-limiting w.r.t.  the compromised robots and
forbidden location. We use the following metric:

\begin{itemize}

\item \emph{Secure stealthy scenario} is the proportion of scenarios that can be
    verified by Alg.~\ref{alg:verify} as secure from the stealthy attacker given a set of
        possible scenarios and is an indicator of how vulnerable the case is to
        stealthy attackers. 

\end{itemize}

A MAPF instance with associated ECBS-computed MAPF plan, co-observation
schedule, and announcement schedule make up an \emph{execution scenario},
or scenario for short.  

\noindent
\paratitle{Bold attacker metrics} In each scenario, we additionally examine the
behavior of a non-stealthy, or bold, attacker that may attempt an attack even if it is not sure the attack will be a forbidden and undetected deviation. The behavior of the bold attacker is deviate to $V_\text{forbidden}$, matching the co-observation schedule as well as possible given the information in $\alpha(t)$ without reasoning about eventual continuations from $\alpha(t)$. We simulate a bold attacker 100 times, each time randomly picking one of
10 different robots to play the attacking role and one of 10 locations
in the grid to be marked as the forbidden location.  We average the metrics over the
scenarios.  For bold attackers, we cannot verify that the plans are
secure, so the CE will attempt to detect, but we cannot guarantee detection.  We
use the following metrics to capture the attacks and their detection: 

\begin{itemize}
\item \emph{Bold attack success} is the 
proportion of simulations where the compromised
robot performs a forbidden deviation 
and is an indicator of how relatively dangerous the compromised robot is in the
set of scenarios. 
\item \emph{Bold detection miss} is the fraction of positive cases where the 
CE reports no anomaly based on our self-report-based detection mechanism
and is an indicator of how many forbidden deviations are missed by the CE. 
\end{itemize}

\subsection{Security Benefit of HoLA}

\begin{figure}[b]
	\centering
	\includegraphics[width=.6\linewidth]{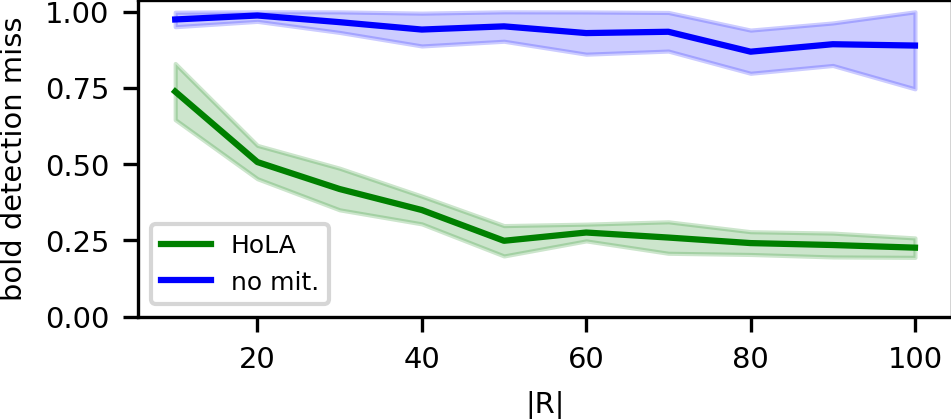}
	\caption{Bold detection miss for the bold attacker when the CE employs either localization-based detection only (no mitigation) or HoLA. In the HoLA case, the CE collects co-observation reports and the announcements are of maximal length that are verified as horizon-limiting.}
	\label{fig:hola-compare}
\end{figure}

The central thesis of this paper is that robot self-reports of localization alone simply do not suffice to detect or prevent malicious
behavior in centralized MRS. As such, the primary contribution of our paper is HoLA, a security measure for the CE leveraging robot co-observations and horizon-limiting announcements. With HoLA, the CE
can compute and release maximal announcements such that there is a guarantee that all stealthy attacks in the system will be prevented. Furthermore, HoLA aids in the detection of non-stealthy attackers in the system by simultaneously making the attack-planning problem more difficult and by gathering the co-observation reports. To demonstrate the necessity of HoLA (\RQ{1}), we consider a bold attacker and we
compare two CE implementations, \begin{enumerate*} \item with localization-based detection only that releases the full MAPF plan to the robots (no mitigation) and \item HoLA: co-observation-based detection where the CE releases maximal-length announcements that have been verified with Alg.~\ref{alg:verify} as preventing all stealthy attacks. \end{enumerate*}

In Fig.~\ref{fig:hola-compare}, we show the miss detection for the  bold attacker as a function of $|R|$, the number of robots in the execution scenario. We observe that when the CE employs no mitigation, essentially all forbidden deviations by the  bold attacker are missed by the CE, highlighting the inadequacy of localization-based detection. Whereas with HoLA, not only are all stealthy attacks provably prevented, the CE misses far fewer bold attacks, ultimately reaching a bold detection miss of just 22\% for $|R|=100$. HoLA consistently outperforms the CE with no mitigation in terms of detection; bold detection miss is lower for larger $|R|$ due to more frequent co-observations in more congested environments. We note that for certain scenarios, the bold attacker is certain to succeed without being detected by HoLA, e.g. when the compromised robots are close to the forbidden zone and far away from other robots.

\subsection{Overhead of Announcement Security Verification}

Our solution proposes that the CE should use Alg.~\ref{alg:verify} to verify
that the chosen $\alpha$ are horizon-limiting MAPF announcements.  The
verification procedure has a computational overhead that depends on the number
of robots, $|R|$. Our scenario
set had instances between 10 and 100 robots with a maximum MAPF length of 70
time steps (the average length is ~49 time steps). Across all scenarios, we
verify each robot in sequence using Alg.~\ref{alg:verify}; it never took longer
than 6 minutes  to terminate. For $|R|=10$, the average time was 48 sec. and
for $|R|=100$, 2.11 min. As the announcements are verified for each robot
independently, the computation is \textit{parallelizable}.  As a point of comparison,
MAPF instances in our benchmark took ECBS up to 1 min. to plan (with
suboptimality bound 1.3), whereas on our 8 core CPU the verification procedure
took up to $6\text{ min.}/8=45\text{ sec.}$
\begin{figure}[t]
	\centering
	\begin{subfigure}[t]{0.282\linewidth}
		\centering
		\includegraphics[width=\linewidth]{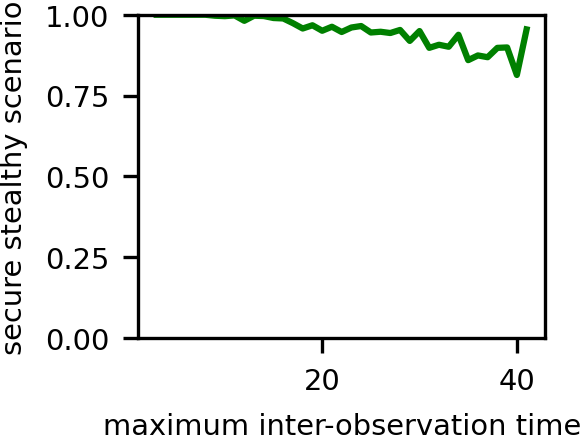}
		\label{fig:inter-obs-secure}
	\end{subfigure}
	\begin{subfigure}[t]{0.282\linewidth}
		\centering
		\includegraphics[width=\linewidth]{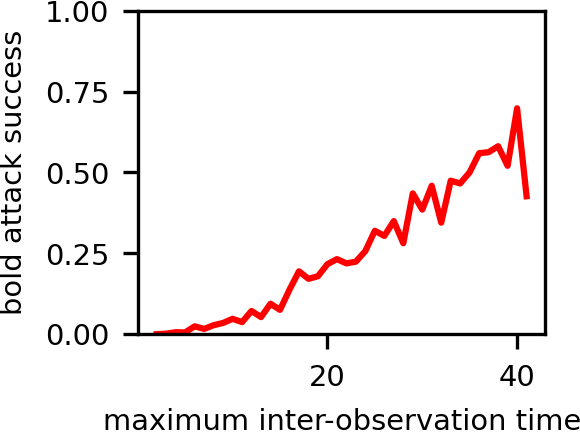}
		\label{fig:inter-obs-attack-success}
	\end{subfigure}
	\caption{
		Attacker success for a bold attacker and secure scenarios for a stealthy attacker
		for {\em general MAPF plans} (i.e. plans not known if they are deviation-detecting because they were not generated
		as such), as a function of the time that robots go unobserved
		(maximum inter-observation time).
	}
	\label{fig:inter-obs}
\end{figure}

\subsection{Security of General MAPF Plans}
We aim to understand what qualities of general MAPF plans contribute to
or detract from the security of the scenario under HoLA (\RQ{3}).  From the perspective of
the attacker, what makes an deviation-detecting plan secure is that there does not
exist an undetected forbidden plan-deviation between consecutive observations
made on the attacker.  We perform the following experiment: we ran attack
scenarios with a simulated bold  attacker where we varied the number
of ahead steps included in announcement  and we
organize the scenarios by the maximum amount of timesteps the attacker has
between consecutive observations, we refer to this as {\em maximum
	inter-observation time}. 

In Fig.~\ref{fig:inter-obs}, we plot the bold attack success and observe that
attackers that have fewer ($<15$) timesteps at most between consecutive
observations have attack success ratio below the 10\% whereas attackers that
have large gaps between consecutive observations ($>25$) have attack success
significantly above average, reaching a attack success of $\sim70\%$ at
maximum inter-observation times of 40 timesteps.  The correlation between
increased maximum inter-observation time and worsened security is also
confirmed by tracking the secure scenarios for the stealthy attacker
verification attempts: we find that attackers that are observed at least every
10 timesteps have 100\% secure scenarios, beyond which point the secure
scenarios trend downward ultimately reaching 81\% for the least-observed
stealthy attackers.

\begin{figure}[t]
	\centering
	\begin{subfigure}[t]{0.282\linewidth}
		\centering
		\includegraphics[width=\linewidth]{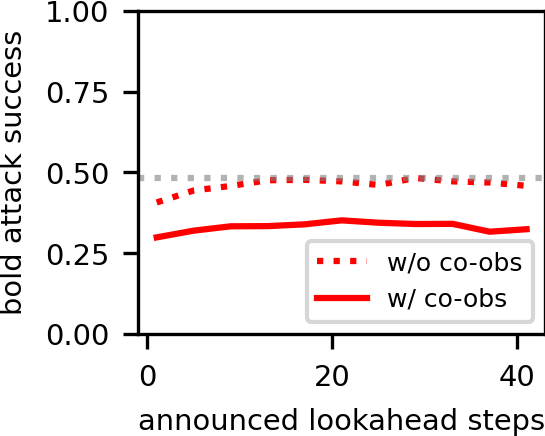}
		\label{fig:benefit-attack-success}
	\end{subfigure}
	\begin{subfigure}[t]{0.282\linewidth}
		\centering
		\includegraphics[width=\linewidth]{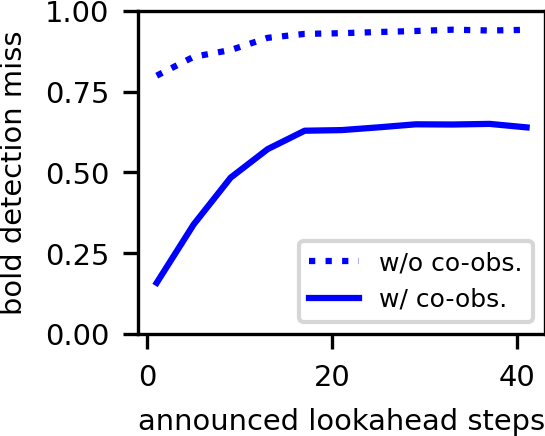}
		\label{fig:benefit-miss}
	\end{subfigure}
	\caption{
		Attack success and miss detection for the bold attacker, with
		co-observations enabled and without co-observations, where the CE uses
		a $(1,k)$-announcement schedule; where k represents the number of lookahead steps.
		In some scenarios the attack is not feasible, resulting in upper bound bold attack success marked with gray dotted line.
	}
	\label{fig:benefit}
\end{figure}

\subsection{Ablation Study: Security Benefit of Robot Co-observations}

In Section~\ref{subsec:coobs-sched} we claimed that robot self-reports
containing only localization information are not sufficient to provide security
guarantees. We support through experimental results this claim (\RQ{4}).  We
consider a bold attacker and we compare two CE implementations, one
with localization-based detection only and one with co-observation-based detection,
where we vary the amount of information available in each announcement, i.e.
how many steps ahead are included in the announcement. 
In order to trigger a detection by the CE in the case where no co-observations
are used,  the compromised robot would need to either collide with a
non-compromised robot or otherwise occupy the location that a non-compromised
robot is meant to occupy.  
For the case with co-observations, a detection would
also occur if the reported co-observations do not match what the CE expected
(see Def.~\ref{def:co-observation-detection}).

In Fig.~\ref{fig:benefit} we show the attack success and the miss detection
for the bold attacker, as a function of $k$, the number of lookahead steps included
in an announcement. We observe that in the situation with minimal
announcements ($k=1$), for the no co-observation setting, the CE has a miss ratio of 80\% whereas with
the robot co-observations present the CE has a miss ratio of
just 16\%. As the announcements become more informative ($k$ larger than 25), the miss
ratio of the no co-observation CE approaches almost 94\% whereas the CE that gathers
co-observations approaches a miss ratio of 64\%. The
bold attacker success decreased from about 46\% to 33\% in the cases 
when no co-observations are used, or when co-observations are used, respectively.

\subsection{Ablation Study: Impact of the Announcement Schedule on Security}

We have argued in Section~\ref{subsec:secure-ann} that announcement schedules impact
the security of a scenario, since more informative announcement schedules decrease the set of plan
deviations that the attacker considers to be possibly forbidden and undetected (\RQ{5}).
However, we have no theoretical guarantee that (1)
Alg.~\ref{alg:verify} is able to   prove more of the less informative scenarios to be horizon-limiting and
(2) that the theoretical increase in attack planning difficulty for bold
attackers under less informative announcements corresponds to a measurable
decrease in attacker success and stealth.

We organize the results
in Fig.~\ref{fig:ann} by the parameter $k$, the number of lookahead steps
in an announcement; by monotonicity of announcements, $(1,k)$-ann.    are less
informative than $(1,k+1)$-ann., and $(k,k)$-ann. are less informative than
$(1,k)$-ann. For the stealthy attacker security verification (see Fig.~\ref{fig:ann}), 
we indeed observe
a negative correlation between $k$ and the secure scenarios: e.g. across all
$(1,k)$-ann. scenarios we have  secure scenarios of approx. 98\% for minimal
($k=1$) announcements, which drops to secure scenarios of approx. 90\% as $k$
increases.  We further report that the density of the agents in the
environment has a large impact on how quickly our ability to verify security
with Alg.~\ref{alg:verify} deteriorates with $k$. As an example, across the
scenarios with few robots, $|R|=10$, the secure scenarios drops to approx. 70\%
whereas for scenarios with $|R|>70$ the secure scenarios does not drop below 95\%.
In addition to supporting that our proposed approach is able to verify the
security of a majority of scenarios w.r.t. stealthy attackers, we note that
increased density of non-compromised robots improves our ability to verify the
security of the system.

\begin{figure}[t]
	\centering
	\begin{subfigure}[t]{0.282\linewidth}
		\centering
		\includegraphics[width=\linewidth]{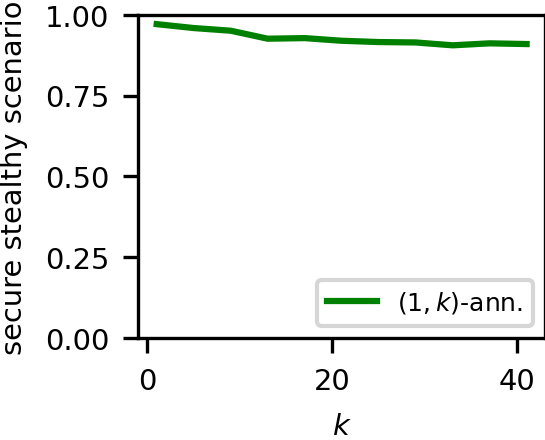}
		\label{fig:secure-scenario}
	\end{subfigure}
	\begin{subfigure}[t]{0.282\linewidth}
		\centering
		\includegraphics[width=\linewidth]{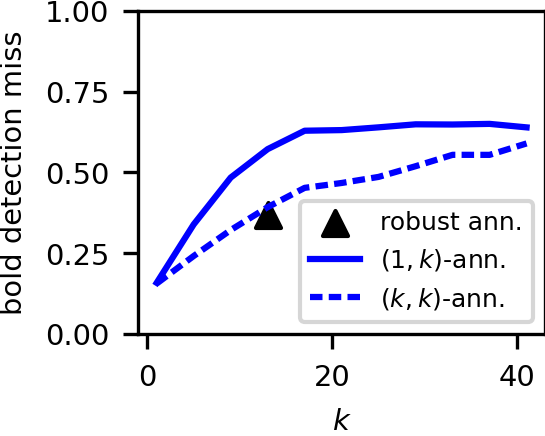}
		\label{fig:detection-miss}
	\end{subfigure}
	\caption{
		Secure scenarios for the stealthy attacker and missed detection for a bold attacker, comparing scenarios where the CE uses $(1,k)$-announcements or $(k,k)$-announcements; announcing the next $k$ steps
		at every time step and announcing the next $k$ steps once every $k$ steps
		respectively. We also plot the detection miss against the average lookahead
		for the robust announcement scheme.
	}
	\label{fig:ann}
\end{figure}

We observe a positive correlation between the
lookahead parameter $k$ of $(p,k)$-announcements and the miss ratio of the CE detections for the bold
attacker (Fig.~\ref{fig:ann}).  Specifically, minimal announcements ($k=1$) correspond to miss
ratios of less than 20\%, but with the most informative announcements tested
the miss ratio approaches up to approx. 60\%. The results indicate that even
though limiting the announcements does not change the set of behaviors
available to bold attackers, in practice the increase in ambiguity the attacker
experiences when choosing a plan deviation has a significant impact on the
stealth of the bold attacker. 

The synthetic announcement classes help to demonstrate the relationship between
information release and security, however the synthetic announcement classes
are not directly comparable to other announcement schemes that may be used in a
typical MRS deployment (\RQ{6}). As a point of comparison, we consider scenarios
where robust announcements are computed from the motion plan as per
\cite{honig_persistent_2019}.  Since the robust announcements have dynamic
prefix lengths that are different for each robot, we measure the average
per-robot prefix length, that is across all plans, times, and robots, the
typical amount of future planning information that is available about a given
robot. We find the average lookahead to be approx. 13 timesteps, and that for
the robust announcement scheme the CE has a miss ratio of approx. 37\%
(Fig.~\ref{fig:ann}). The results indicate that from a security
perspective, the robust announcement scheme is quite similar to the
$(k,k)$-announcement class (since $(13,13)$-announcements have a similar miss
ratio).

\section{Related Work}
\label{sec:relwork}
\paratitle{Patrolling} Most relevant to our work is the use of robots in a physical security context has been
considered in the context of adversarial multi-robot patrolling (MRP) games,
where a multi-robot system should be programmed to maximize detections of  intruders attempting to
penetrate to a forbidden zone~\cite{agmon2008multi}.
Adversaries could have zero, partial, or full system
knowledge~\cite{agmon2011multi}; the intrusion detection could be centralized
or decentralized~\cite{fagiolini2007decentralized,fagiolini2008consensus}; and
the detector might use fixed sensors in addition to patrolling
robots~\cite{kim2008multi}.
MRP can be differentiated from our work in
several ways. In MRP, patrolling robots' only objective is patrolling, the
patrolling robots are assumed trustworthy, and the intruder is an outsider -- 
whereas in our setting the robots' primary objective is servicing application
tasks, the robots may be compromised, and the ``intruder'' (the robot
attempting to penetrate the forbidden zone) is an insider.

\noindent
\paratitle{Robust MAPF} Prior work on MAPF has proposed different announcement
schedules in order to 
allow for more flexible re-planning in case of agent failures or motion delays
\cite{honig2019persistent,atzmon_robust_2020}, or provide fault-tolerant
robot planning~\cite{yang2011fault,arrichiello2015observer}.
In this paper we focus on mitigating plan-deviation attacks, thus our announcement
schedule is focused on incremental disclosure of knowledge for security purposes. Recent work in multi-robot surveillance has considered how compromised robots can effectively deny service, e.g. in~\cite{liu2021distributed}, resilience to compromised robots is cast as a robust task scheduling problem.

\noindent
\paratitle{Security for robotic applications}	
Several works study robot and multi-agent security; a survey 
is presented in~\cite{bijani2014review,yaacoub2021robotics}.
Robot cyber security is 
	analyzed at the communication-level by Bicchi~\emph{et.~al.}~\cite{bicchi2008heterogeneous}
	and Renganathan and Summers~\cite{renganathan2017spoof};
	considered with a human-in-the-loop by Portugal~\emph{et.~al.}~\cite{portugal2017role};
	and discussed broadly by Morante, Victores, and Balaguer~\cite{morante2015cryptobotics}.
Insecurities arising from interactions between robots and the physical environment were studied
in a number of works such as vulnerabilities in robotic arms of the type used in factory
assembly lines~\cite{quarta2017experimental}, 
vulnerabilities in robot sensors~\cite{choi2020software}, and vulnerabilities in actuators ~\cite{guo2018roboads}.
Some works show how attackers can exploit software vulnerabilities in the
Robot Operating System (ROS) for attacks and propose corresponding security enhancements
\cite{dieber2016application,rivera2019ros}. The proposed
defenses are focused on the software and not the robotic applications themselves.

\noindent
\paratitle{Location-based attacks in routing protocols}
The work in \cite{shoukry2018sybil}
presents Sybil attack-resilient traffic estimation and routing algorithm that uses information
from sensing infrastructure and the dynamics and proximities of vehicles.
Other works build attack-resilient network protocols by exploiting physical properties of the system~\cite{gil2017guaranteeing}. The setting in this case is very different from our problem 
where there is a central entity, CE, that is doing the planning and detection and the robots are 
constrained in how they move.

\vspace{-1ex}
\section{Conclusion}
\label{sec:concl}

In this paper we focused on the problem of mitigating
	\emph{plan-deviation attacks} 
	with robot co-observations and incremental plan release.
The attacker has two goals: first, to move toward a forbidden
zone, and second, to remain undetected by the central entity.
We leverage co-observation to mitigate the ability of the attacker to lie about
its location; and we limit the size of the incremental plan announcements so
that the attacker has limited ability to confidently plan ahead.  We describe
two types of attackers -- ``stealthy'', and ``bold'' -- based on their desire
to remain undetected or not. We prove that our solution prevents attacks for
a set of stealthy attackers. For bold attackers we show
experimentally that our solution significantly increases the detection of the
attacks. Our solution also has a small overhead making it practical for sets of tens to hundreds of robots.

\balance

\bibliographystyle{unsrtnat}
\bibliography{bib/IEEEConfFull,bib/IEEEFull,bib/OtherFull,bib/main}


\begin{thebibliography}{36}


\ifx \showCODEN    \undefined \def \showCODEN     #1{\unskip}     \fi
\ifx \showDOI      \undefined \def \showDOI       #1{#1}\fi
\ifx \showISBNx    \undefined \def \showISBNx     #1{\unskip}     \fi
\ifx \showISBNxiii \undefined \def \showISBNxiii  #1{\unskip}     \fi
\ifx \showISSN     \undefined \def \showISSN      #1{\unskip}     \fi
\ifx \showLCCN     \undefined \def \showLCCN      #1{\unskip}     \fi
\ifx \shownote     \undefined \def \shownote      #1{#1}          \fi
\ifx \showarticletitle \undefined \def \showarticletitle #1{#1}   \fi
\ifx \showURL      \undefined \def \showURL       {\relax}        \fi
\providecommand\bibfield[2]{#2}
\providecommand\bibinfo[2]{#2}
\providecommand\natexlab[1]{#1}
\providecommand\showeprint[2][]{arXiv:#2}

\bibitem[abi(2019)]%
        {abirob}
 \bibinfo{year}{2019}\natexlab{}.
\newblock \bibinfo{title}{Consumer Robotics is a Market in Transition; Smart
  Home Will be at the Heart of the Change}.
\newblock
  \bibinfo{howpublished}{\url{https://www.abiresearch.com/press/consumer-robotics-market-transition-smart-home-will-be-heart-change/}}.
\newblock
\newblock
\shownote{Accessed 4 April 2021.}.


\bibitem[iro(2020)]%
        {irobotcorp}
 \bibinfo{year}{2020}\natexlab{}.
\newblock \bibinfo{title}{History}.
\newblock
  \bibinfo{howpublished}{\url{https://www.irobot.com/about-irobot/company-information/history}}.
\newblock
\newblock
\shownote{Accessed 4 April 2021.}.


\bibitem[spo(2021)]%
        {spotrob}
 \bibinfo{year}{2021}\natexlab{}.
\newblock \bibinfo{title}{Spot}.
\newblock \bibinfo{howpublished}{\url{https://www.bostondynamics.com/spot}}.
\newblock
\newblock
\shownote{Accessed 4 April 2021.}.


\bibitem[Agmon et~al\mbox{.}(2011)]%
        {agmon2011multi}
\bibfield{author}{\bibinfo{person}{Noa Agmon}, \bibinfo{person}{Gal~A Kaminka},
  {and} \bibinfo{person}{Sarit Kraus}.} \bibinfo{year}{2011}\natexlab{}.
\newblock \showarticletitle{Multi-robot adversarial patrolling: facing a
  full-knowledge opponent}.
\newblock \bibinfo{journal}{\emph{Journal of Artificial Intelligence Research}}
   \bibinfo{volume}{42} (\bibinfo{year}{2011}), \bibinfo{pages}{887--916}.
\newblock


\bibitem[Agmon et~al\mbox{.}(2008)]%
        {agmon2008multi}
\bibfield{author}{\bibinfo{person}{Noa Agmon}, \bibinfo{person}{Sarit Kraus},
  {and} \bibinfo{person}{Gal~A Kaminka}.} \bibinfo{year}{2008}\natexlab{}.
\newblock \showarticletitle{Multi-robot perimeter patrol in adversarial
  settings}. In \bibinfo{booktitle}{\emph{IEEE International Conference on
  Robotics and Automation}}. IEEE, \bibinfo{pages}{2339--2345}.
\newblock


\bibitem[Arrichiello et~al\mbox{.}(2015)]%
        {arrichiello2015observer}
\bibfield{author}{\bibinfo{person}{Filippo Arrichiello},
  \bibinfo{person}{Alessandro Marino}, {and} \bibinfo{person}{Francesco
  Pierri}.} \bibinfo{year}{2015}\natexlab{}.
\newblock \showarticletitle{Observer-based decentralized fault detection and
  isolation strategy for networked multirobot systems}.
\newblock \bibinfo{journal}{\emph{IEEE Transactions on Control Systems
  Technology}} \bibinfo{volume}{23}, \bibinfo{number}{4}
  (\bibinfo{year}{2015}), \bibinfo{pages}{1465--1476}.
\newblock


\bibitem[Atzmon et~al\mbox{.}(2020)]%
        {atzmon_robust_2020}
\bibfield{author}{\bibinfo{person}{Dor Atzmon}, \bibinfo{person}{Roni Stern},
  \bibinfo{person}{Ariel Felner}, \bibinfo{person}{Glenn Wagner},
  \bibinfo{person}{Roman Barták}, {and} \bibinfo{person}{Neng-Fa Zhou}.}
  \bibinfo{year}{2020}\natexlab{}.
\newblock \showarticletitle{Robust {Multi}-{Agent} {Path} {Finding} and
  {Executing}}.
\newblock \bibinfo{journal}{\emph{Journal of Artificial Intelligence Research}}
   \bibinfo{volume}{67} (\bibinfo{date}{March} \bibinfo{year}{2020}),
  \bibinfo{pages}{549--579}.
\newblock


\bibitem[Azadeh et~al\mbox{.}(2019)]%
        {azadeh2019robotized}
\bibfield{author}{\bibinfo{person}{Kaveh Azadeh}, \bibinfo{person}{Ren{\'e}
  De~Koster}, {and} \bibinfo{person}{Debjit Roy}.}
  \bibinfo{year}{2019}\natexlab{}.
\newblock \showarticletitle{Robotized and automated warehouse systems: Review
  and recent developments}.
\newblock \bibinfo{journal}{\emph{Transportation Science}}
  \bibinfo{volume}{53}, \bibinfo{number}{4} (\bibinfo{year}{2019}),
  \bibinfo{pages}{917--945}.
\newblock


\bibitem[Barer et~al\mbox{.}(2014)]%
        {barer_suboptimal_2014}
\bibfield{author}{\bibinfo{person}{Max Barer}, \bibinfo{person}{Guni Sharon},
  \bibinfo{person}{Roni Stern}, {and} \bibinfo{person}{Ariel Felner}.}
  \bibinfo{year}{2014}\natexlab{}.
\newblock \showarticletitle{Suboptimal {Variants} of the {Conflict}-{Based}
  {Search} {Algorithm} for the {Multi}-{Agent} {Pathfinding} {Problem}}.
\newblock \bibinfo{journal}{\emph{Proceedings of the 7th Annual Symposium on
  Combinatorial Search}} (\bibinfo{year}{2014}).
\newblock


\bibitem[Bicchi et~al\mbox{.}(2008)]%
        {bicchi2008heterogeneous}
\bibfield{author}{\bibinfo{person}{Antonio Bicchi}, \bibinfo{person}{Antonio
  Danesi}, \bibinfo{person}{Gianluca Dini}, \bibinfo{person}{Silvio La~Porta},
  \bibinfo{person}{Lucia Pallottino}, \bibinfo{person}{Ida~M Savino}, {and}
  \bibinfo{person}{Riccardo Schiavi}.} \bibinfo{year}{2008}\natexlab{}.
\newblock \showarticletitle{Heterogeneous wireless multirobot system}.
\newblock \bibinfo{journal}{\emph{IEEE robotics \& automation magazine}}
  \bibinfo{volume}{15}, \bibinfo{number}{1} (\bibinfo{year}{2008}),
  \bibinfo{pages}{62--70}.
\newblock


\bibitem[Bijani and Robertson(2014)]%
        {bijani2014review}
\bibfield{author}{\bibinfo{person}{Shahriar Bijani} {and}
  \bibinfo{person}{David Robertson}.} \bibinfo{year}{2014}\natexlab{}.
\newblock \showarticletitle{A review of attacks and security approaches in open
  multi-agent systems}.
\newblock \bibinfo{journal}{\emph{Artificial Intelligence Review}}
  \bibinfo{volume}{42}, \bibinfo{number}{4} (\bibinfo{year}{2014}),
  \bibinfo{pages}{607--636}.
\newblock


\bibitem[Choi et~al\mbox{.}(2020)]%
        {choi2020software}
\bibfield{author}{\bibinfo{person}{Hongjun Choi}, \bibinfo{person}{Sayali
  Kate}, \bibinfo{person}{Yousra Aafer}, \bibinfo{person}{Xiangyu Zhang}, {and}
  \bibinfo{person}{Dongyan Xu}.} \bibinfo{year}{2020}\natexlab{}.
\newblock \showarticletitle{Software-based Realtime Recovery from Sensor
  Attacks on Robotic Vehicles}. In \bibinfo{booktitle}{\emph{23rd International
  Symposium on Research in Attacks, Intrusions and Defenses}}.
  \bibinfo{pages}{349--364}.
\newblock


\bibitem[Choset et~al\mbox{.}(2005)]%
        {choset2005principles}
\bibfield{author}{\bibinfo{person}{Howie~M Choset}, \bibinfo{person}{Kevin~M
  Lynch}, \bibinfo{person}{Seth Hutchinson}, \bibinfo{person}{George Kantor},
  \bibinfo{person}{Wolfram Burgard}, \bibinfo{person}{Lydia Kavraki},
  \bibinfo{person}{Sebastian Thrun}, {and} \bibinfo{person}{Ronald~C Arkin}.}
  \bibinfo{year}{2005}\natexlab{}.
\newblock \bibinfo{booktitle}{\emph{Principles of robot motion: theory,
  algorithms, and implementation}}.
\newblock \bibinfo{publisher}{{MIT} press}.
\newblock


\bibitem[Dieber et~al\mbox{.}(2016)]%
        {dieber2016application}
\bibfield{author}{\bibinfo{person}{Bernhard Dieber}, \bibinfo{person}{Severin
  Kacianka}, \bibinfo{person}{Stefan Rass}, {and} \bibinfo{person}{Peter
  Schartner}.} \bibinfo{year}{2016}\natexlab{}.
\newblock \showarticletitle{Application-level security for {ROS}-based
  applications}. In \bibinfo{booktitle}{\emph{International Conference on
  Intelligent Robots and Systems}}. IEEE, \bibinfo{pages}{4477--4482}.
\newblock


\bibitem[Fagiolini et~al\mbox{.}(2008)]%
        {fagiolini2008consensus}
\bibfield{author}{\bibinfo{person}{Adriano Fagiolini}, \bibinfo{person}{Marco
  Pellinacci}, \bibinfo{person}{Gianni Valenti}, \bibinfo{person}{Gianluca
  Dini}, {and} \bibinfo{person}{Antonio Bicchi}.}
  \bibinfo{year}{2008}\natexlab{}.
\newblock \showarticletitle{Consensus-based distributed intrusion detection for
  multi-robot systems}. In \bibinfo{booktitle}{\emph{2008 IEEE International
  Conference on Robotics and Automation}}. IEEE, \bibinfo{pages}{120--127}.
\newblock


\bibitem[Fagiolini et~al\mbox{.}(2007)]%
        {fagiolini2007decentralized}
\bibfield{author}{\bibinfo{person}{Adriano Fagiolini}, \bibinfo{person}{Gianni
  Valenti}, \bibinfo{person}{Lucia Pallottino}, \bibinfo{person}{Gianluca
  Dini}, {and} \bibinfo{person}{Antonio Bicchi}.}
  \bibinfo{year}{2007}\natexlab{}.
\newblock \showarticletitle{Decentralized intrusion detection for secure
  cooperative multi-agent systems}. In \bibinfo{booktitle}{\emph{2007 46th IEEE
  Conference on Decision and Control}}. IEEE, \bibinfo{pages}{1553--1558}.
\newblock


\bibitem[Gil et~al\mbox{.}(2017)]%
        {gil2017guaranteeing}
\bibfield{author}{\bibinfo{person}{Stephanie Gil}, \bibinfo{person}{Swarun
  Kumar}, \bibinfo{person}{Mark Mazumder}, \bibinfo{person}{Dina Katabi}, {and}
  \bibinfo{person}{Daniela Rus}.} \bibinfo{year}{2017}\natexlab{}.
\newblock \showarticletitle{Guaranteeing spoof-resilient multi-robot networks}.
\newblock \bibinfo{journal}{\emph{Autonomous Robots}} \bibinfo{volume}{41},
  \bibinfo{number}{6} (\bibinfo{year}{2017}), \bibinfo{pages}{1383--1400}.
\newblock


\bibitem[Guo et~al\mbox{.}(2018)]%
        {guo2018roboads}
\bibfield{author}{\bibinfo{person}{Pinyao Guo}, \bibinfo{person}{Hunmin Kim},
  \bibinfo{person}{Nurali Virani}, \bibinfo{person}{Jun Xu},
  \bibinfo{person}{Minghui Zhu}, {and} \bibinfo{person}{Peng Liu}.}
  \bibinfo{year}{2018}\natexlab{}.
\newblock \showarticletitle{RoboADS: Anomaly detection against sensor and
  actuator misbehaviors in mobile robots}. In \bibinfo{booktitle}{\emph{2018
  48th Annual IEEE/IFIP international conference on dependable systems and
  networks (DSN)}}. IEEE, \bibinfo{pages}{574--585}.
\newblock


\bibitem[H{\"o}nig(2021)]%
        {libmultirobotplanning2021}
\bibfield{author}{\bibinfo{person}{Wolfgang H{\"o}nig}.}
  \bibinfo{year}{2021}\natexlab{}.
\newblock \bibinfo{title}{libMultiRobotPlanning}.
\newblock
\newblock
\urldef\tempurl%
\url{https://github.com/whoenig/libMultiRobotPlanning}
\showURL{%
\tempurl}


\bibitem[H{\"o}nig et~al\mbox{.}(2019)]%
        {honig2019persistent}
\bibfield{author}{\bibinfo{person}{Wolfgang H{\"o}nig}, \bibinfo{person}{Scott
  Kiesel}, \bibinfo{person}{Andrew Tinka}, \bibinfo{person}{Joseph~W Durham},
  {and} \bibinfo{person}{Nora Ayanian}.} \bibinfo{year}{2019}\natexlab{}.
\newblock \showarticletitle{Persistent and robust execution of {MAPF} schedules
  in warehouses}.
\newblock \bibinfo{journal}{\emph{{IEEE} Robotics and Automation Letters}}
  \bibinfo{volume}{4}, \bibinfo{number}{2} (\bibinfo{year}{2019}),
  \bibinfo{pages}{1125--1131}.
\newblock


\bibitem[Honig et~al\mbox{.}(2019)]%
        {honig_persistent_2019}
\bibfield{author}{\bibinfo{person}{Wolfgang Honig}, \bibinfo{person}{Scott
  Kiesel}, \bibinfo{person}{Andrew Tinka}, \bibinfo{person}{Joseph~W. Durham},
  {and} \bibinfo{person}{Nora Ayanian}.} \bibinfo{year}{2019}\natexlab{}.
\newblock \showarticletitle{Persistent and {Robust} {Execution} of {MAPF}
  {Schedules} in {Warehouses}}.
\newblock \bibinfo{journal}{\emph{IEEE Robotics and Automation Letters}}
  \bibinfo{volume}{4}, \bibinfo{number}{2} (\bibinfo{date}{April}
  \bibinfo{year}{2019}), \bibinfo{pages}{1125--1131}.
\newblock
\showISSN{2377-3766, 2377-3774}
\urldef\tempurl%
\url{https://doi.org/10.1109/LRA.2019.2894217}
\showDOI{\tempurl}


\bibitem[Kim et~al\mbox{.}(2008)]%
        {kim2008multi}
\bibfield{author}{\bibinfo{person}{Ji~Min Kim}, \bibinfo{person}{Jeong~Sik
  Choi}, {and} \bibinfo{person}{Beom~Hee Lee}.}
  \bibinfo{year}{2008}\natexlab{}.
\newblock \showarticletitle{Multi-agent coordinated motion planning for
  monitoring and controlling the observed space in a security zone}.
\newblock \bibinfo{journal}{\emph{IFAC Proceedings Volumes}}
  \bibinfo{volume}{41}, \bibinfo{number}{2} (\bibinfo{year}{2008}),
  \bibinfo{pages}{1679--1684}.
\newblock


\bibitem[Li et~al\mbox{.}(2021)]%
        {li_lifelong_2021}
\bibfield{author}{\bibinfo{person}{Jiaoyang Li}, \bibinfo{person}{Andrew
  Tinka}, \bibinfo{person}{Scott Kiesel}, \bibinfo{person}{Joseph~W Durham},
  \bibinfo{person}{T~K~Satish Kumar}, {and} \bibinfo{person}{Sven Koenig}.}
  \bibinfo{year}{2021}\natexlab{}.
\newblock \showarticletitle{Lifelong {Multi}-{Agent} {Path} {Finding} in
  {Large}-{Scale} {Warehouses}}.
\newblock \bibinfo{journal}{\emph{Proceedings of the AAAI Conference on
  Artificial Intelligence (AAAI)}} (\bibinfo{year}{2021}).
\newblock


\bibitem[Liu et~al\mbox{.}(2021)]%
        {liu2021distributed}
\bibfield{author}{\bibinfo{person}{Jun Liu}, \bibinfo{person}{Lifeng Zhou},
  \bibinfo{person}{Pratap Tokekar}, {and} \bibinfo{person}{Ryan~K Williams}.}
  \bibinfo{year}{2021}\natexlab{}.
\newblock \showarticletitle{Distributed resilient submodular action selection
  in adversarial environments}.
\newblock \bibinfo{journal}{\emph{IEEE Robotics and Automation Letters}}
  \bibinfo{volume}{6}, \bibinfo{number}{3} (\bibinfo{year}{2021}),
  \bibinfo{pages}{5832--5839}.
\newblock


\bibitem[Morante et~al\mbox{.}(2015)]%
        {morante2015cryptobotics}
\bibfield{author}{\bibinfo{person}{Santiago Morante}, \bibinfo{person}{Juan~G
  Victores}, {and} \bibinfo{person}{Carlos Balaguer}.}
  \bibinfo{year}{2015}\natexlab{}.
\newblock \showarticletitle{Cryptobotics: Why robots need cyber safety}.
\newblock \bibinfo{journal}{\emph{Frontiers in Robotics and AI}}
  \bibinfo{volume}{2} (\bibinfo{year}{2015}), \bibinfo{pages}{23}.
\newblock


\bibitem[Portugal et~al\mbox{.}(2017)]%
        {portugal2017role}
\bibfield{author}{\bibinfo{person}{David Portugal}, \bibinfo{person}{Samuel
  Pereira}, {and} \bibinfo{person}{Micael~S Couceiro}.}
  \bibinfo{year}{2017}\natexlab{}.
\newblock \showarticletitle{The role of security in human-robot shared
  environments: A case study in {ROS}-based surveillance robots}. In
  \bibinfo{booktitle}{\emph{2017 26th {IEEE} International Symposium on Robot
  and Human Interactive Communication ({RO-MAN})}}. IEEE,
  \bibinfo{pages}{981--986}.
\newblock


\bibitem[Quarta et~al\mbox{.}(2017)]%
        {quarta2017experimental}
\bibfield{author}{\bibinfo{person}{Davide Quarta}, \bibinfo{person}{Marcello
  Pogliani}, \bibinfo{person}{Mario Polino}, \bibinfo{person}{Federico Maggi},
  \bibinfo{person}{Andrea~Maria Zanchettin}, {and} \bibinfo{person}{Stefano
  Zanero}.} \bibinfo{year}{2017}\natexlab{}.
\newblock \showarticletitle{An experimental security analysis of an industrial
  robot controller}. In \bibinfo{booktitle}{\emph{2017 IEEE Symposium on
  Security and Privacy (SP)}}. IEEE, \bibinfo{pages}{268--286}.
\newblock


\bibitem[Renganathan and Summers(2017)]%
        {renganathan2017spoof}
\bibfield{author}{\bibinfo{person}{Venkatraman Renganathan} {and}
  \bibinfo{person}{Tyler Summers}.} \bibinfo{year}{2017}\natexlab{}.
\newblock \showarticletitle{Spoof resilient coordination for distributed
  multi-robot systems}. In \bibinfo{booktitle}{\emph{2017 International
  Symposium on Multi-Robot and Multi-Agent Systems (MRS)}}. IEEE,
  \bibinfo{pages}{135--141}.
\newblock


\bibitem[Rivera et~al\mbox{.}(2019)]%
        {rivera2019ros}
\bibfield{author}{\bibinfo{person}{Sean Rivera}, \bibinfo{person}{Sofiane
  Lagraa}, \bibinfo{person}{Cristina Nita-Rotaru}, \bibinfo{person}{Sheila
  Becker}, {and} \bibinfo{person}{Radu State}.}
  \bibinfo{year}{2019}\natexlab{}.
\newblock \showarticletitle{{ROS}-defender: {SDN}-based security policy
  enforcement for robotic applications}. In \bibinfo{booktitle}{\emph{2019
  {IEEE} Security and Privacy Workshops ({SPW})}}. IEEE,
  \bibinfo{pages}{114--119}.
\newblock


\bibitem[Shoukry et~al\mbox{.}(2018)]%
        {shoukry2018sybil}
\bibfield{author}{\bibinfo{person}{Yasser Shoukry}, \bibinfo{person}{Shaunak
  Mishra}, \bibinfo{person}{Zutian Luo}, {and} \bibinfo{person}{Suhas
  Diggavi}.} \bibinfo{year}{2018}\natexlab{}.
\newblock \showarticletitle{Sybil attack resilient traffic networks: A
  physics-based trust propagation approach}. In \bibinfo{booktitle}{\emph{2018
  ACM/IEEE 9th International Conference on Cyber-Physical Systems (ICCPS)}}.
  IEEE, \bibinfo{pages}{43--54}.
\newblock


\bibitem[Simon(2019)]%
        {amazonrob}
\bibfield{author}{\bibinfo{person}{Matt Simon}.}
  \bibinfo{year}{2019}\natexlab{}.
\newblock \bibinfo{title}{Inside the {A}mazon Warehouse Where Humans and
  Machines Become One}.
\newblock
  \bibinfo{howpublished}{\url{https://www.wired.com/story/amazon-warehouse-robots/}}.
\newblock
\newblock
\shownote{Accessed 4 April 2021.}.


\bibitem[Stern et~al\mbox{.}(2019)]%
        {stern_multi-agent_2019}
\bibfield{author}{\bibinfo{person}{Roni Stern}, \bibinfo{person}{Nathan~R.
  Sturtevant}, \bibinfo{person}{Ariel Felner}, \bibinfo{person}{Sven Koenig},
  \bibinfo{person}{Hang Ma}, \bibinfo{person}{Thayne~T. Walker},
  \bibinfo{person}{Jiaoyang Li}, \bibinfo{person}{Dor Atzmon},
  \bibinfo{person}{Liron Cohen}, \bibinfo{person}{T.~K.~Satish Kumar},
  \bibinfo{person}{Roman Barták}, {and} \bibinfo{person}{Eli Boyarski}.}
  \bibinfo{year}{2019}\natexlab{}.
\newblock \showarticletitle{Multi-{Agent} {Pathfinding}: {Definitions},
  {Variants}, and {Benchmarks}}. In \bibinfo{booktitle}{\emph{Twelfth {Annual}
  {Symposium} on {Combinatorial} {Search}}}.
\newblock


\bibitem[Wardega et~al\mbox{.}(2019)]%
        {wardega_resilience_2019}
\bibfield{author}{\bibinfo{person}{Kacper Wardega}, \bibinfo{person}{Roberto
  Tron}, {and} \bibinfo{person}{Wenchao Li}.} \bibinfo{year}{2019}\natexlab{}.
\newblock \showarticletitle{Resilience of multi-robot systems to physical
  masquerade attacks}.
\newblock \bibinfo{journal}{\emph{Proceedings - 2019 IEEE Symposium on Security
  and Privacy Workshops, SPW 2019}} (\bibinfo{year}{2019}),
  \bibinfo{pages}{120--125}.
\newblock


\bibitem[Yaacoub et~al\mbox{.}(2021)]%
        {yaacoub2021robotics}
\bibfield{author}{\bibinfo{person}{Jean-Paul~A Yaacoub},
  \bibinfo{person}{Hassan~N Noura}, \bibinfo{person}{Ola Salman}, {and}
  \bibinfo{person}{Ali Chehab}.} \bibinfo{year}{2021}\natexlab{}.
\newblock \showarticletitle{Robotics cyber security: Vulnerabilities, attacks,
  countermeasures, and recommendations}.
\newblock \bibinfo{journal}{\emph{International Journal of Information
  Security}} (\bibinfo{year}{2021}), \bibinfo{pages}{1--44}.
\newblock


\bibitem[Yang et~al\mbox{.}(2011)]%
        {yang2011fault}
\bibfield{author}{\bibinfo{person}{Hao Yang}, \bibinfo{person}{Marcel
  Staroswiecki}, \bibinfo{person}{Bin Jiang}, {and} \bibinfo{person}{Jianye
  Liu}.} \bibinfo{year}{2011}\natexlab{}.
\newblock \showarticletitle{Fault tolerant cooperative control for a class of
  nonlinear multi-agent systems}.
\newblock \bibinfo{journal}{\emph{Systems \& control letters}}
  \bibinfo{volume}{60}, \bibinfo{number}{4} (\bibinfo{year}{2011}),
  \bibinfo{pages}{271--277}.
\newblock


\bibitem[Yu and LaValle(2013)]%
        {yu2013structure}
\bibfield{author}{\bibinfo{person}{Jingjin Yu} {and} \bibinfo{person}{Steven~M
  LaValle}.} \bibinfo{year}{2013}\natexlab{}.
\newblock \showarticletitle{Structure and intractability of optimal multi-robot
  path planning on graphs}. In \bibinfo{booktitle}{\emph{Twenty-Seventh AAAI
  Conference on Artificial Intelligence}}.
\newblock


\end{thebibliography}

\end{document}